\documentclass[11pt,twoside]{article}
\usepackage{graphicx,hyperref,color}
\usepackage[utf8]{inputenc}
\usepackage[margin=1in]{geometry}
\usepackage{marvosym}
\usepackage{cite}
\usepackage{verbatim}
\usepackage{amsmath,amssymb,amsthm}
\usepackage{amsmath,amssymb}
\usepackage{chngcntr}
\usepackage[stable]{footmisc}
\usepackage{float}
\usepackage{url}
\usepackage{xspace}
\usepackage{paralist}
\usepackage{multienum}

\usepackage{thmtools}
\usepackage{thm-restate}
\usepackage{hyperref}
\usepackage{cleveref}

\usepackage{authblk}

\newcommand{\EMPH}[1]{\emph{#1}}

\newtheorem{theorem}{Theorem}[section]

\newtheorem{lemma}[theorem]{Lemma}

\newtheorem{observation}[theorem]{Observation}

\newcommand{\OPT}{OPT}
\newcommand{\terminals}{$\{(s_{1}, t_{1}), \ldots, (s_{k}, t_{k})\}$}
\newcommand{\opts}{$\{\opt{1}, \opt{2}, \ldots, \opt{k}\}$}
\newcommand{\spath}[1]{P_{#1}}
\newcommand{\opt}[1]{Q_{#1}}
\newcommand{\region}[1]{R_{#1}}
\newcommand{\border}[2]{B_{#1,#2}}
\newcommand{\len}[1]{\ell(#1)}

\newcommand{\eric}{Colin~de Verdi{\`e}re}

\title{Towards single face shortest vertex-disjoint paths in undirected planar graphs}

\author{Glencora Borradaile}
\author{Amir Nayyeri}
\author{Farzad Zafarani}
\affil{School of Electrical Engineering and Computer Science \\ Oregon State University\\
	\{glencora, nayyeria, zafaranf\}@eecs.oregonstate.edu}

\begin{document}
\maketitle

\begin{abstract}
Given $k$ pairs of terminals \terminals\ in a graph $G$, the min-sum $k$ vertex-disjoint paths problem is to find a collection $\{\opt{1}, \opt{2}, \ldots, \opt{k}\}$ of vertex-disjoint paths with minimum total length, where $\opt{i}$ is an $s_i$-to-$t_i$ path between $s_i$ and $t_i$.  We consider the problem in planar graphs, where little is known about computational tractability, even in restricted cases. Kobayashi and Sommer propose a polynomial-time algorithm for $k \le 3$ in undirected planar graphs assuming all terminals are adjacent to at most two faces.
\eric\ and Schrijver give a polynomial-time algorithm when all the sources are on the boundary of one face and all the sinks are on the boundary of another face and ask about the existence of a polynomial-time algorithm provided all terminals are on a common face.  

We make progress toward \eric\ and Schrijver's open question by giving an $O(kn^5)$ time algorithm for undirected planar graphs when \terminals\ are in counter-clockwise order on a common face.
\end{abstract}

\newpage
\section{Introduction}

Given $k$ pairs of terminals \terminals, the \EMPH{$k$ vertex-disjoint} paths problem asks for a set of $k$ disjoint paths \opts, in which $\opt{i}$ is a path between $s_i$ and $t_i$ for all $1\leq i\leq k$.
As a special case of the multi-commodity flow problem, computing vertex disjoint paths has found several applications, for example in VLSI design\cite{kramer1984complexity}, or network routing \cite{ogier1993distributed,srinivas2005finding}.  
It is one of Karp's NP-hard problems \cite{karp1974} even for undirected planar graphs if $k$ is part of the input \cite{middendorf1993disjPathComplexity}.  
However, there are polynomial time algorithms if $k$ is a constant for general undirected graphs~\cite{robertson1995graph, Kawarabayashi2010shorterProofGraphMinorAlg}.
In general directed graphs, the $k$-vertex-disjoint paths problem is NP-hard even for $k = 2$~\cite{fortune1980directedSubGraphHom} but is fixed parameter tractable with respect to parameter $k$ in directed planar graphs~\cite{schrijver1994finding, cygan2013planarDirDisjFPT}.

Surprisingly, much less is known for the optimization variant of the problem, \EMPH{minimum-sum $k$ vertex-disjoint paths problem ($k$-min-sum)}, where a set of disjoint paths with \emph{minimum total length} is desired. 
For example, the $2$-min-sum problem and the $4$-min-sum problem are open in directed and undirected planar graphs, respectively, even when the terminals are on a common face; neither polynomial-time algorithms nor hardness results are known for these problems~\cite{kobayashi2010shortest}. Bjorklund and Husfeldt gave a randomized polynomial time algorithm for the min-sum two vertex-disjoint paths problem in general undirected graphs \cite{bjorklund2014shortest}. Kobayashi and Sommer provide a comprehensive list of similar open problems (Table 2~\cite{kobayashi2010shortest}).

One of a few results in this context is due to \eric~and Schrijver~\cite{verdiere2011shortest}: a polynomial time algorithm for the $k$-min-sum problem in a (directed or undirected) planar graph, given all sources are on one face and all sinks are on another face~\cite{verdiere2011shortest}.  In the same paper, they ask about the existence of a polynomial time algorithm provided all the terminals (sources and sinks) are on a common face. If the sources and sinks are ordered so that they are in the order $s_1, s_2, \ldots, s_k, t_k, t_{k-1}, \ldots, t_1$ around the boundary, then the $k$-min-sum problem can be solved by finding a min-cost flow from $s_1, s_2, \ldots, s_k$ to $t_k, t_{k-1}, \ldots, t_1$.  For $k \le 3$ in undirected planar graphs with the terminals in arbitrary order around the common face, Kobayashi and Sommer give an $O(n^4\log n)$ algorithm~\cite{kobayashi2010shortest}\footnote{Kobayashi and Sommer also describe algorithms for the case where terminals are on two different faces, and $k = 3$.}.
In this paper, we give the first polynomial-time algorithm for an arbitrary number of terminals on the boundary of a common face, which we call $F$, so long as the terminals alternate along the boundary.  Formally, we prove:
\begin{theorem}
There exists an $O(kn^5)$ time algorithm to solve the $k$-min-sum problem, provided that the terminals $s_1, t_1, s_2, t_2, \ldots, s_k, t_k$ are in counter-clockwise order on the boundary of the graph.
\end{theorem}

\paragraph{Definitions and assumptions} We use standard notation for graphs and planar graphs.  For simplicity, we assume that the terminal vertices are distinct.  One could imagine allowing $t_i = s_{i+1}$; our algorithm can be easily modified to handle this case.  We also assume that the shortest path between any two vertices of the input graph is unique as it significantly simplifies the presentation of our result; this assumption can be enforced using a perturbation technique~\cite{MVV87}.

\section{Preliminaries}

\paragraph{Walks and paths.} Let $G = (V, E)$ be a graph, and let $H$ be a subgraph of $G$.  We use $V(H)$ and $E(H)$ to denote the vertex set and the edge set of $H$, respectively.  For $U \subseteq V$, we use \EMPH{$G[U]$} to denote the subgraph of $G$ induced by $U$, whose vertex set is $U$ and whose edge set is all edges of $G$ with both endpoints in $U$.  A \EMPH{walk} $W = (v_{1}, v_{2}, \ldots, v_{h})$ in $G$ is a sequence of vertices such that for all $1\leq i < h$, we have $(v_i, v_{i+1}) \in E$.  For any $v_i, v_j \in W$, \EMPH{$W[v_i, v_j]$} is the (sub-)walk $(v_i, v_{i+1}, \ldots, v_j)$.
A \EMPH{path} is a walk with no repeated vertices.  Sometimes, we view a walk as its set of edges, and use set operations on walks.  For example, given two walks $W_1$ and $W_2$, we define $W = W_1\oplus W_2$ to be their symmetric difference when it is clear from the context that $W$ is a walk.  Given a length function $\ell:E\rightarrow V$, the length of $W$ is denoted by \EMPH{$\len{W}$}, and it is $\sum_{i=1}^{h-1}{\ell(v_i, v_{i+1})}$.

\paragraph{Planarity.}
An \EMPH{embedding} of a graph $G = (V,E)$ into the Euclidean plane is a mapping of vertices of $G$ into different points of $\mathbb{R}^2$, and edges of $G$ into internally disjoint simple curves in $\mathbb{R}^2$ such that the endpoints of the image of $(u,v) \in E$ are the images of vertices $u \in V$ and $v \in V$.  If such an embedding exists then $G$ is a \EMPH{planar graph}.  A \EMPH{plane graph} is a planar graph and an embedding of it.  The \EMPH{faces} of a plane graph $G$ are the maximal connected components of $\mathbb{R}^2$ that are disjoint from the image of $G$.  If $G$ is connected it contains only one unbounded face.  This is called the \EMPH{outer face} of $G$ and we denote the boundary of $G$ by $\partial G$.
Let $W$ be a walk and let $W^o$ be the set of edges that are used an odd number of times in $W$.  
$W^o$ is a collection of simple cycles.  For any point $p\in \mathbb{R}^2$, we say that $p$ is inside $W$ if $p$ is on the image of $W$ in the plane, or $c$ is inside an odd number of cycles of $W^o$.
When it is clear from the context, we use the same notation to refer to the subgraph composed of the vertices and edges whose images are completely inside $W$.

\section{Structural properties}
In this section, we present fundamental properties of the optimum solution that we exploit in our algorithm.
To simplify the exposition, we search for pairwise disjoint walks rather than simple paths and refer to a set of pairwise disjoint walks connecting corresponding pair of terminals as a \emph{feasible} solution.
Indeed, in an optimal solution, the walks are simple paths.

Let $\{\opt{1}, \opt{2}, \ldots, \opt{k}\}$ be an optimal solution, where $\opt{i}$ is a $s_i$-to-$t_i$ path and let $\{\spath{1}, \spath{2}, \ldots, \spath{k}\}$ be the set of shortest paths, where $\spath{i}$ is the $s_i$-to-$t_i$ shortest path.
These shortest paths together with the boundary of the graph, $\partial G$, define internally disjoint regions of the plane.  Specifically, we define $\region{i}$ to be the subset of the plane bounded by the cycle $\spath{i} \cup C_i$, where $C_i$ is the $s_i$-to-$t_i$ subpath of $\partial G$ that does not contain other terminal vertices.
The following lemmas constrain the behavior of the optimal paths.

\begin{lemma}
\label{lem:sp_in_region}
For all $1\leq i\leq k$, the path $\opt{i}$ is inside $\region{i}$.
\end{lemma}
\begin{proof}
Suppose, to derive a contradiction, that $\opt{i} \nsubseteq \region{i}$.
So, there is a vertex $v\in \opt{i}\backslash \region{i}$.  Let $p_v$ be the maximal path containing $v$ that is internally disjoint from $\region{i}$.  Let $(x,y)$ be endpoints of $p_v$, and observe that $x,y \in \spath{i}$.  Also, by uniqueness of shortest paths, we have $\len{\spath{i}[x,y]} < \len{\opt{i}[x,y]}$.  Thus, $\opt{i}' = \opt{i} \oplus \opt{i}[x,y] \oplus \spath{i}[x,y]$ is shorter than $\opt{i}$.  

It remains to show that $\opt{i}' \cap \opt{j} = \emptyset$ for all $j \neq i$, $1\leq j\leq k$.  But, by the construction, $\opt{i}'$ is inside $C_i\cup \opt{i}$, and all terminals other than $s_i$ and $t_i$ are outside $C_i\cup \opt{i}$.  So, by Jordan curve theorem, for any $\opt{j}$ to intersect $\opt{i}'$, it has to intersect $\opt{i}$, too.  Thus, $\{\opt{1}, \ldots, \opt{k}\}\backslash \opt{i} \cup \opt{i}'$ is a shorter solution than the optimum, which is a contradiction.
\end{proof}

\noindent We take the vertices of $\spath{i}$ and $\opt{i}$ to be ordered along these paths from $s_i$ to $t_i$.

\begin{lemma}
\label{lem:pqSameOrder}
For $u,v \in \opt{i} \cap \spath{i}$, $u$ precedes $v$ in $\spath{i}$ if and only if $u$ precedes $v$ in $\opt{i}$.
\end{lemma}
\begin{proof}
Suppose, to derive a contradiction, that $u$ and $v$ have different orders on $\spath{i}[s_i, t_i]$ and $\opt{i}[s_i, t_i]$, and assume, without loss of generality, that $u$ precedes $v$ in $\spath{i}[s_i, t_i]$, but $v$ precedes $u$ in $\opt{i}[s_i, t_i]$.  So, $\opt{i}[s_i, t_i]$ can be decomposed into three subpaths (1) $\gamma_1$ is a $s_i$-to-$v$ path, (2) $\gamma_2$ is a $v$-to-$u$ path, and (3) $\gamma_3$ is a $u$-to-$t_i$ path.  If $\gamma_1$ contains $u$ then $\opt{i}$ is not a simple path, visiting $u$ at least twice. Otherwise, The Jordan curve theorem implies that $\gamma_2$ or $\gamma_3$ must intersect $\gamma_1$.  Again $\opt{i}$ is not simple, so, it is not a walk in the optimum solution.
\end{proof}

\begin{figure}[tbh]
  \centering
    \includegraphics[height=1.2in]{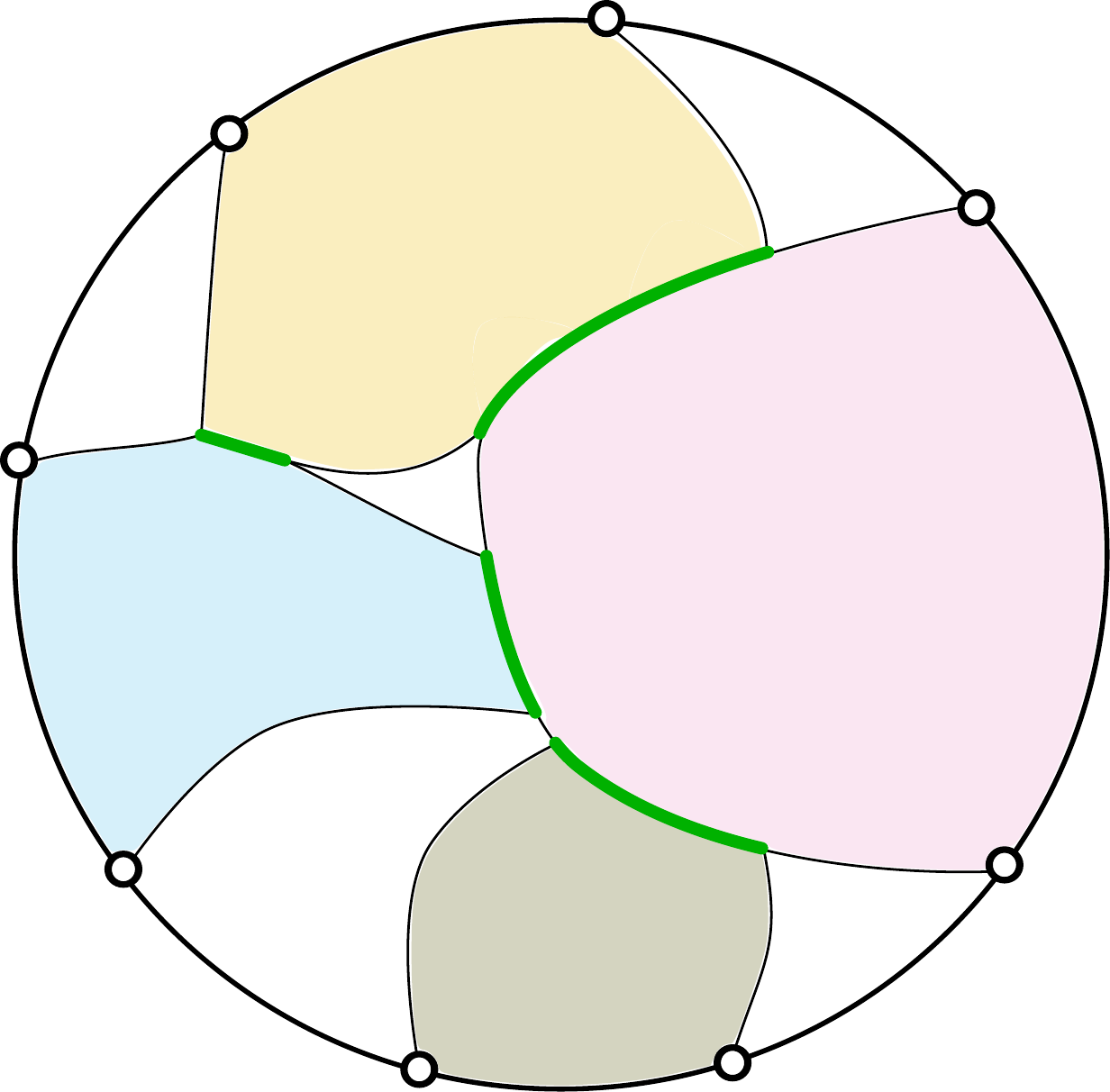}\qquad\qquad
    \includegraphics[height=1.2in]{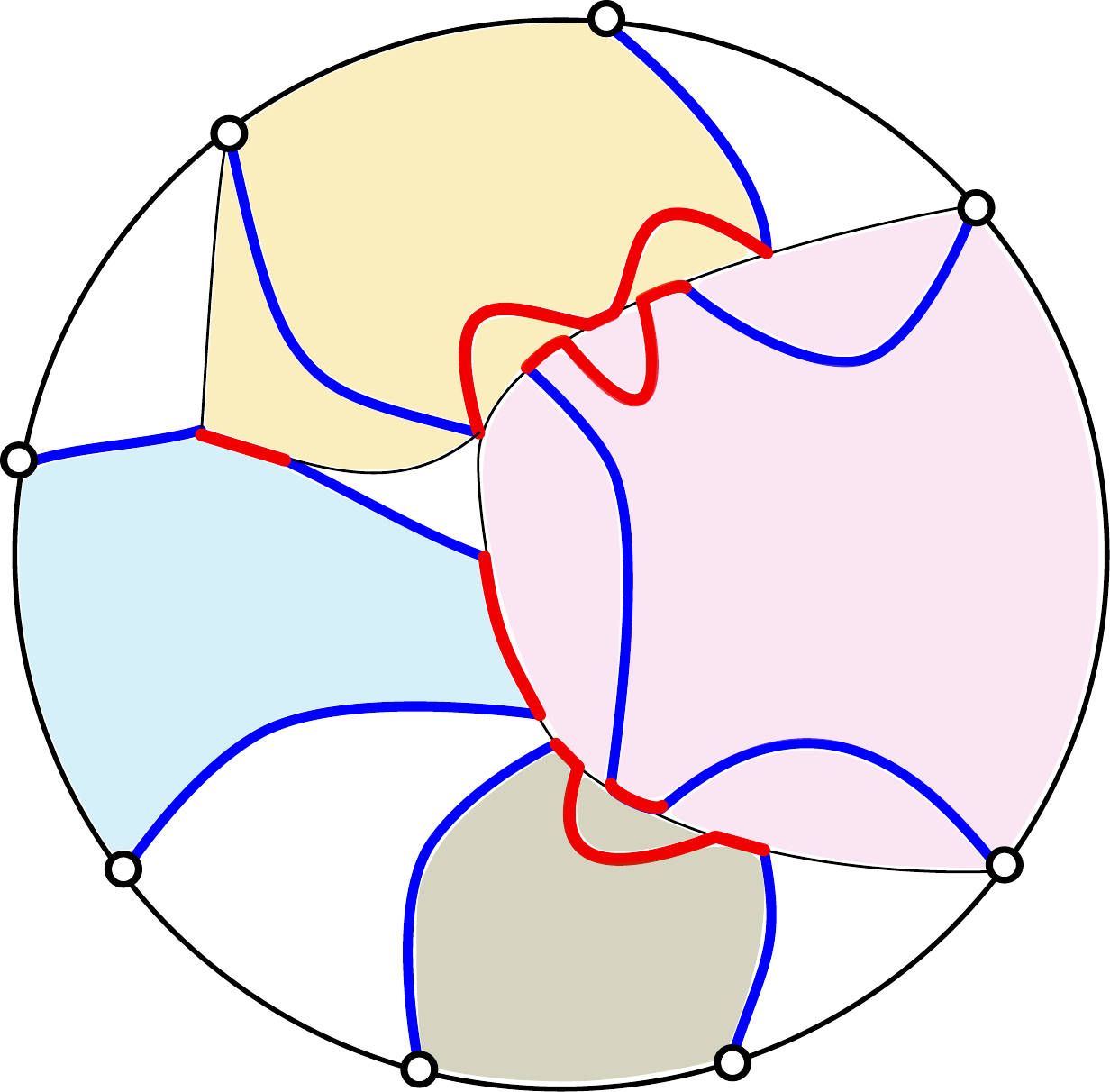}
      \caption{(left) A $4$-min sum instance; regions are shaded and borders are green. (right) A feasible solution; Type~I and Type~II subpaths are blue and red, respectively.}
  \label{fig:4-Disj}
\end{figure}

We call $\region{i} \cap \region{j}$ the \emph{border} of $\region{i}$ and $\region{j}$ and denote it $\border{i}{j}$.   Note that a border can be a single vertex. Since we assume shortest paths are unique, $\border{i}{j}$ is a single (shortest) path.
Figure~\ref{fig:4-Disj} illustrates borders for a $4$-min-sum instance.  The following lemma bounds the total number of borders.

\begin{lemma}
\label{lem:bounded_number_borders}
There are $O(k)$ border paths. 
\end{lemma}
\begin{proof}
Let $\mathcal{R} = (V_{\mathcal{R}}, E_{\mathcal{R}})$ be the graph whose vertices correspond to regions in $G$, and there is an edge between two vertices of $V_{\mathcal{R}}$ if their corresponding regions in $G$ share a border.  Let $H = G[\partial F \cup \spath{1} \cup \ldots \cup \spath{k}]$, and observe that $\mathcal{R}$ is a subgraph of the planar dual of $H$.  Thus, $\mathcal{R}$ is planar.

Since there is a bijection between $V_{\mathcal{R}}$ and the set of regions of $G$, we have $|V_{\mathcal{R}}| = k$.
Additionally, there is a bijection between $E_{\mathcal{R}}$ and borders in $G$.  Because, $\mathcal{R}$ is planar we conclude that the number of borders in $G$ is $|E_{\mathcal{R}}| = O(V_{\mathcal{R}}) = O(k)$.
\end{proof}

Consider a region $\region{i}$ and consider the borders along $P_i$, $\border{i}{i_1}, \border{i}{i_2}, \ldots, \border{i}{i_t}$.  Observe that the intersections of the regions $R_{i_1}, R_{i_2}, \ldots, R_{i_t}$ with $\partial G$ must be in a {\em clockwise} order.  Let $\iota_1, \ldots, \iota_\ell$ be the subsequence of $i_1, \ldots, i_t$ of indices to regions that intersect $\opt{i}$.  For $j \in \{\iota_1, \ldots, \iota_\ell\}$, let $x_j$ and $y_j$ be the first and last vertex of $\opt{i}$ in $\border{i}{j}$.  Additionally, define $y_0 = s_i$ and $x_{\ell+1} = t_i$. We partition $\opt{i}$ into a collection of subpaths of two types as follows.
\begin{enumerate}[Type~I :]
\item For $h = 0, \ldots, \ell$, $\opt{i}[y_h, x_{h+1}]$ is a {\em Type~I} subpath in region $\region{i}$.
\item For $h = 1,\ldots, \ell-1$, $\opt{i}[x_h, y_{h}]$ is a {\em Type~II} subpath in region $\region{i}$.  We say that $\opt{i}[x_h, y_{h}]$ is on the border $\border{i}{j}$ containing $x_h$ and $y_{h}$.
\end{enumerate}
By this definition, all Type~I paths are internally disjoint from all borders. By Lemma~\ref{lem:pqSameOrder}, each Type~II path is internally disjoint from all borders except possibly the border that contains its endpoints, with which it may have several intersecting points.  See Figure~\ref{fig:4-Disj} for an illustration of Type~I~and~II paths.

The following lemma demonstrates a key property of Type~I paths, implying that (given their endpoints) they can be computed efficiently via a shortest path computation:
\begin{lemma} 
\label{lem:alpha_opt}
Let $\alpha$ be a Type~I subpath in region $\region{i}$.  
Then $\alpha$ is the shortest path between its endpoints in $\region{i}$ that is internally disjoint from all borders.
\end{lemma}
\begin{proof}
Let $u,v\in V(G)$ be the endpoints of $\alpha$, and let $\alpha'$ be the shortest $u$-to-$v$ path in $\region{i}$ that is internally disjoint from all borders.
Also, let $\opt{i}' = \opt{i} \oplus \alpha \oplus \alpha'$.  
The path $\alpha'$ has the same endpoints as $\alpha$, and it is internally disjoint from all $\region{j}$ if $j \neq i$.
Also, by Lemma~\ref{lem:sp_in_region}, for each $1\leq j\leq k$, $\opt{j}$ is inside $\region{j}$.
Therefore, $\alpha'$ is disjoint from $\opt{j}$ for all $1\leq j\leq k$ and $j\neq i$.
Thus, $\{\opt{1}, \ldots, \opt{k}\}\backslash \opt{i} \cup \opt{i}'$ is a set of $k$ pairwise disjoint walks, with total length $\OPT - \len{\alpha} + \len{\alpha'}$ where $\OPT$ is the total length of the optimal paths.
So, $\OPT - \len{\alpha} + \len{\alpha'} \geq \OPT$, which implies $\len{\alpha'}  \geq \len{\alpha}$.
Therefore $\alpha$ must be a shortest $(u,v)$ path in $\region{i}$ that avoids boundaries.  In fact, uniqueness of shortest paths implies $\alpha = \alpha'$.
\end{proof}

A Type~II path has a similar property if it is the only Type~II path on the border that contains its endpoints. The proof of the following lemma is almost exactly the same as the previous proof.
\begin{lemma} 
\label{lem:beta_opt0}
Let $\beta$ be a Type~II subpath in region $\region{i}$ on border $\border{i}{j}$.  Suppose there is no Type~II path on $\border{i}{j}$ inside $\region{j}$.
Then $\beta$ is the subpath of $\border{i}{j}$ between its endpoints.  
\end{lemma}

The following lemma reveals a relatively more sophisticated structural property of Type~II paths on shared borders.

\begin{lemma}
\label{lem:beta_opt}
Let $\beta$ and $\gamma$ be  Type~II subpaths in $\region{i}$ and $\region{j}$ on $\border{i}{j}$, respectively, let $x_i$ and $y_i$ be the endpoints of $\beta$, and let $x_j$ and $y_j$ be the endpoints of $\gamma$.
Then, $\{\beta, \gamma\}$ is the pair of paths with \emph{minimum total length} with the following properties:
\begin{enumerate}[(1)]
\item $\beta$ is an $x_i$-to-$y_i$ path inside $\region{i}$, and it is internally disjoint from all borders except possibly $\border{i}{j}[x_i, y_i]$.
\item $\gamma$ is an $x_j$-to-$y_j$ path inside $\region{j}$, and it is internally disjoint from all borders except possibly $\border{i}{j}[x_j, y_j]$.  
\end{enumerate}
\end{lemma}
\begin{proof}
Properties (1) and (2) of $\beta$ and $\gamma$ are implied by the definition of Type~II paths and because they are internally disjoint from all borders except $\border{i}{j}$.  It remains to show that their total length is minimum.

Let $(\beta',\gamma')$ be the pair of paths with minimum total length that has properties (1) and (2). 
Let $\opt{i}' = \opt{i}\oplus \beta \oplus \beta'$, and let $\opt{j}' = \opt{j} \oplus \gamma \oplus \gamma'$.

By construction, $\beta'$ is internally disjoint from all borders except (possibly) $\border{i}{j}[x_i, y_i]$.  Additionally, the endpoints of $\beta'$ are the same as the endpoints of $\beta$.
Therefore, intersection points of $\beta'$ with borders of $\region{i}$ that are not in $\beta$ are all in $\border{i}{j}[x_i, y_i]$.
Similarly, intersection points of $\gamma'$ with borders of $\region{j}$ that are not in $\beta'$ are all on $\border{i}{j}[x_j, y_j]$.
It immediately follows that $\opt{i}'$ and $\opt{h}$ are disjoint for any $h\in \{1,2,\ldots, k\}\backslash \{j\}$.
Similarly, $\opt{j}'$ and $\opt{h}$ are disjoint for any $h\in \{1,2,\ldots, k\}\backslash \{i\}$.

Furthermore, $\opt{i}'$ and $\opt{j}'$ are disjoint by their construction.  Thus, $\{\opt{1}, \ldots, \opt{k}\}\backslash \{\opt{i},\opt{j}\} \cup \{\opt{i}', \opt{j}'\}$ is a set of $k$ pairwise disjoint walks connecting the terminals.  The total length of this new set is $OPT - \len{\beta} - \len{\gamma} + \len{\beta'} + \len{\gamma'} \geq OPT$.  Consequently, $\len{\beta} + \len{\gamma} \leq \len{\beta'} + \len{\gamma'}$, in fact, $\len{\beta} + \len{\gamma} = \len{\beta'} + \len{\gamma'}$.  Thus, $(\alpha, \beta)$ are a pair of path with minimum total length that has properties (1) and (2).
\end{proof}

\section{Algorithmic toolbox} 

In this section, we describe algorithms to compute paths of Type~I~and~II for given endpoints.
These algorithms are key ingredients of our strongly polynomial time algorithm described in the next section.
More directly, they imply an $n^{O(k)}$ time algorithm via enumerating the endpoints, which is sketched at the end of this section.

Each Type I path can be computed in linear time using the algorithms of Henzinger et al.~\cite{henzinger1997planarShortestPaths}; they can also be computed in bulk in $O(n \log n)$ time using the multiple-source shortest path algorithm of Klein~\cite{klein2005multiSourceShPaths} (although other parts of our algorithms dominate the shortest path computation).
Similarly, a Type~II path on a border $\border{i}{j}$ can be computed in linear time provided it is the only path on $\border{i}{j}$.
Computing pairs of Type~II paths on a shared border is slightly more challenging.  
To achieve this, our algorithm reduces this problem into a $2$-min sum problem that can be solved in linear time via a reduction to the minimum-cost flow problem.  
The following lemma is implicit in the paper of Kobayashi and Sommer~\cite{kobayashi2010shortest}. 

\begin{lemma}
\label{lem:2minSum}
There exists a linear time algorithm to solve the $2$-min sum problem on an undirected planar graph, provided the terminals are on the outer face.
\end{lemma}
\begin{proof}
Consider an instance of the $2$-min sum problem with terminals $\{(s_1, t_1), (s_2, t_2)\}$. This problem reduces to a minimum-cost flow problem as follows.
Because of the symmetry for terminals in undirected graphs, we can assume that $s_1$ and $s_2$ are next to each other on the outer face:  there is a $s_1$-to-$s_2$ path on the boundary of the outer face that does not intersect $\{t_1, t_2\}$. We make the graph directed by replacing each undirected edge $\{u,v\}$ with edges $(u, v)$ and $(v, u)$. For edge $(u,v)$ in the directed graph, we assign its length using length function $\ell(u,v)$. For every vertex $v$ in $G$, we split it into two vertices $v_{1}$ and $v_{2}$ and connect them with a zero length edge that has unit capacity. For each edge $(u, v)$, we connect $u$ to $v_{1}$, and for each edge $(v, u)$, we connect $v_{2}$ to $u$. We introduce a dummy source vertex $d_1$, with two edges $(d_1, s_1)$, and $(d_1, s_2)$ of unit capacity and zero length.  Also, we introduce a dummy sink vertex $d_2$, with edges $(t_1, d_2)$ and $(t_2, d_2)$ of unit capacity and zero length.  We assign capacity one to all other edges.  The lengths of the other edges are specified by the length function $\ell$ of $G$. Since, $s_1$ and $s_2$ (also $t_1$ and $t_2$) are next to each other, the graph remains planar after adding $u$, $v$, and their incident edges.
Now, it is straight forward to see that our $2$-min sum problem is equivalent to a minimum cost $u$-to-$v$ flow problem of value $2$.
This minimum cost flow problem in turn reduces to two shortest path computations from $u$ to $v$, which can be done in linear time~\cite{henzinger1997planarShortestPaths}.
\end{proof}

We reduce the computation of Type II paths to $2$-min sum.  The following lemma is a slightly stronger form of this reduction, which finds application in our strongly polynomial time algorithm. 

\begin{lemma}
 \label{lem:2minSumplus}
Let $\region{i}$ and $\region{j}$ be two regions with border $\border{i}{j}$ and let $x_i, y_i \in  \spath{i}$ and $x_j, y_j \in \spath{j}$.  A pair of paths $(\beta, \gamma)$ with total minimum length and with the following properties can be computed in linear time.
\begin{enumerate}
\item $\beta$ is an $x_i$-to-$y_i$ path inside $\region{i}$, and it is internally disjoint from all borders except possibly $\spath{i}[x_i, y_i] \cap \border{i}{j}$.
\item $\gamma$ is an $x_j$-to-$y_j$ path inside $\region{j}$, and it is internally disjoint from all borders except possibly $\spath{j}[x_j, y_j] \cap \border{i}{j}$.  
\end{enumerate}
\end{lemma}
\begin{proof}
Let the graph $H$ be the induced subgraph by the vertices of $G$ inside $\region{i} \cup \region{j}$.
We obtain $H'$ from $H$ by performing the following operations:
\begin{enumerate}
\item deleting all vertices of $H$ that belong to borders other than $\border{i}{j}$,
\item deleting all edges in $\region{i}$ that are incident to $\border{i}{j}$ but not incident to $\spath{i}[x_i, y_i] \cap \border{i}{j}$, and
\item deleting all edges in $\region{j}$ that are incident to $\border{i}{j}$ but not incident to $\spath{j}[x_j, y_j] \cap \border{i}{j}$.
\end{enumerate}
Note that $H'$ is not necessarily connected.  For an illustration of $H$ and $H'$, see Figure~\ref{fig:mincost}.

\begin{figure}[tbh]
  \centering
    \includegraphics[height=1.5in]{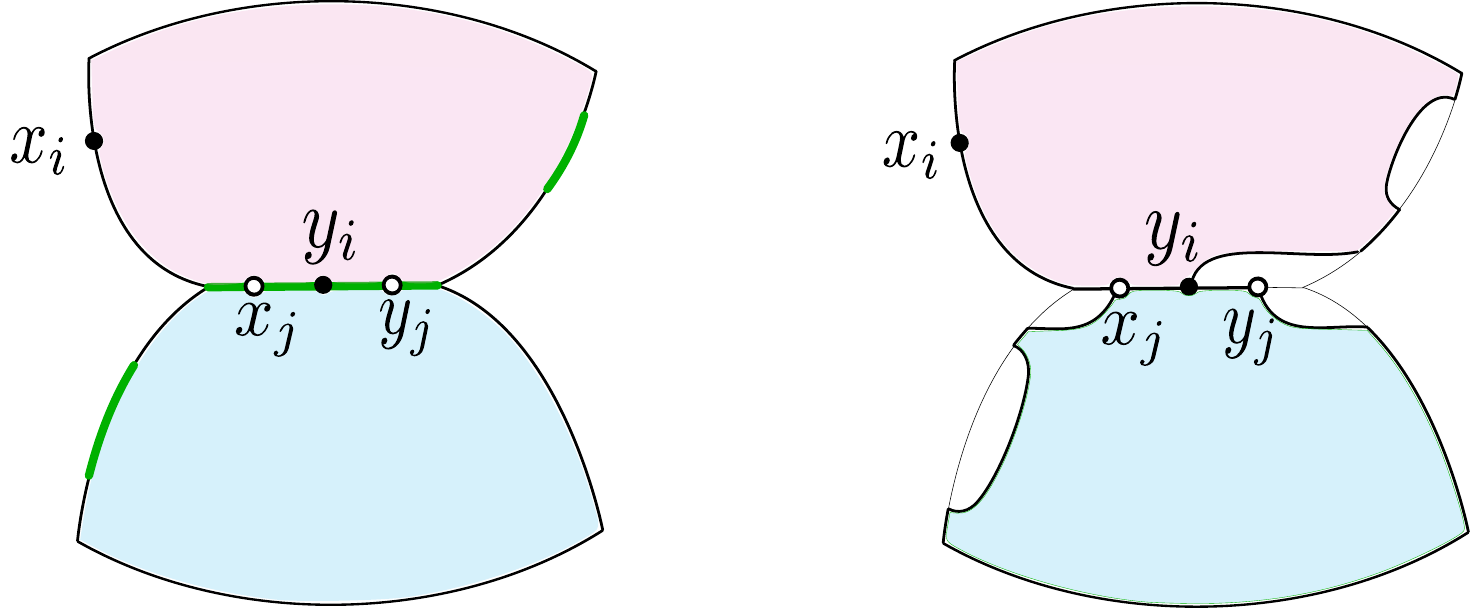}
      \caption{(left) The graph $H$ induced by vertices in $\region{i}\cup\region{j}$; regions are shaded and borders are green. (right) The graph $H'$ defined in the proof of Lemma~\ref{lem:2minSumplus}}
  \label{fig:mincost}
\end{figure}

Since $\beta$ and $\gamma$ are intact in $H'$, they can be computed in $H'$ instead of $G$.
Furthermore, observe that $x_i, y_i, x_j, y_j$ are on the boundary of $H'$.  If $H'$ is disconnected, then $\beta$ and $\gamma$ are a shortest paths in their own connected components.  So they can be computed in linear time using the algorithm of Henzinger et al.~\cite{henzinger1997planarShortestPaths}.

So, suppose $H'$ is connected, and observe that $x_i, y_i, x_j, y_j$ are on the boundary of the outer face of $H'$.  By Lemma~\ref{lem:2minSum}, there is a linear time algorithm to compute a pair of disjoint paths of minimum total length between corresponding terminals.

Let $\beta'$ and $\gamma'$ be $x_i$-to-$y_i$ path and $x_j$-to-$y_j$ path computed via solving a $2$-min sum problem. 
It remains to prove that $\beta'$ and $\gamma'$ have properties (1) and (2) of the lemma.
That is $\beta' \in H'\cap \region{i}$, and $\gamma' \in H'\cap\region{j}$.  This can be done through a replacing path argument similar to Lemma~\ref{lem:sp_in_region}, as $\border{i}{j}$ (so, any subpath of it) is a shortest path.
\end{proof}

\subsection{An $n^{O(k)}$ time algorithm}

The properties of Type~I~and~II paths imply a na\"ive $n^{O(k)}$ time algorithm, which we sketch here.
An optimal solution defines the endpoints of Type~I and Type~II paths, so we can simply enumerate over which borders contain endpoints of Type~I and~II paths and then enumerate over the choices of the endpoints.  
Consequently, there are zero, two, or four (not necessarily distinct) endpoints of Type~I and~II paths on $\border{i}{j}$, or
\[
1 + {\len{\border{i}{j}} \choose 2} + {\len{\border{i}{j}} \choose 4}
\]
possibilities, which is $O(n^4)$ since $\len{\border{i}{j}} = O(n)$.  Since there are $O(k)$ borders (Lemma~\ref{lem:bounded_number_borders}), there are $n^{O(k)}$ endpoints to guess.  Given the set of endpoints, we compute a feasible solution composed of the described Type~I and~II paths or determines that no such solution exists.  Since Type~I and~II paths can be computed in polynomial time, the overall 
algorithm runs in $n^{O(k)}$ time.


\section{A fully polynomial time algorithm}

We give an $O(kn^5)$-time algorithm via dynamic programming over the regions.  For two regions $R_i$ and $R_j$ that have a shared border $B_{i,j}$, $R_i$ and $R_j$ separate the terminal pairs into two sets: those terminals $s_\ell,t_\ell$ for $\ell = i+1,\ldots, j-1$ and $s_m,t_m$ for $m = j+1, \ldots, k, 1, \ldots, i-1$ (for $i < j$).  Any $s_\ell$-to-$t_\ell$ path that is in region $R_\ell$ cannot touch any $s_m$-to-$t_m$ path that is in region $R_m$ since $R_\ell$ and $R_m$ are vertex disjoint.  Therefore any influence the $s_\ell$-to-$t_\ell$ path has on the $s_m$-to-$t_m$ path occurs indirectly through the $s_i$-to-$t_i$ and $s_j$-to-$t_j$ paths.  Our dynamic program is indexed by the shared borders $B_{i,j}$ and pairs of vertices on (a subpath of) $P_i$ and (a subpath of) $P_j$; the vertices on $P_i$ and $P_j$ will indicate a {\em last point} on the boundary of $R_i$ and $R_j$ that a (partial) feasible solution uses. 

We use a tree to order the shared borders for processing by the dynamic program. Since there are $O(k)$ borders (Lemma~\ref{lem:bounded_number_borders}), the dynamic programming table has $O(k n^2)$ entries.  We formally define the dynamic programming table below and show how to compute each entry in $O(n^3)$ time.

\subsection{Dynamic programming tree}
Let ${\cal R} = \{\region{i}\}_{i=1}^k$ and $\cal B$ be the set of all borders between all pairs of regions.
We assume, without loss of generality, that $\cal R$ is connected, otherwise we  split the problem into independent subproblems, one for each connected component of $\cal R$.

We define a graph $T$ (that we will argue is a tree) whose edges are the shared borders $\cal B$ between the regions $\mathcal R$.
Two distinct borders $\border{i}{j}$ and $B_{h,\ell}$ are incident in this graph if there is an endpoint $x$ of $\border{i}{j}$ and $y$ of $B_{h,\ell}$  that are connected by an $x$-to-$y$ curve in $\mathbb{R}^2 \setminus F$ that does not touch any region $\mathcal R$ except at its endpoints $x$ and $y$;  see Figure~\ref{fig:T}.  Note that this curve may be trivial (i.e.\ $x = y$).  The vertices of $T$ (in a non-degenerate instance) correspond one-to-one with components of $\mathbb{R}^2 \setminus (F \cup {\cal R})$ (plus some additional trivial components if three or more regions intersect at a point), or {\em non-regions}. The edges of $T$ cannot form a cycle, since by the Jordan Curve Theorem this would define a disk that is disjoint from $\partial F$;  an edge $B_{i,j}$ in the cycle bounds two regions, one of which would be contained by the disk, contradicting that each region shares a boundary with $\partial F$.  Therefore $T$ is indeed a tree.   We use an embedding of $T$ that is derived naturally from the embedding of $G$ according to this construction.  We use this tree to guide the dynamic program.

\begin{figure}[ht]
  \centering
        \includegraphics[height=1.2in]{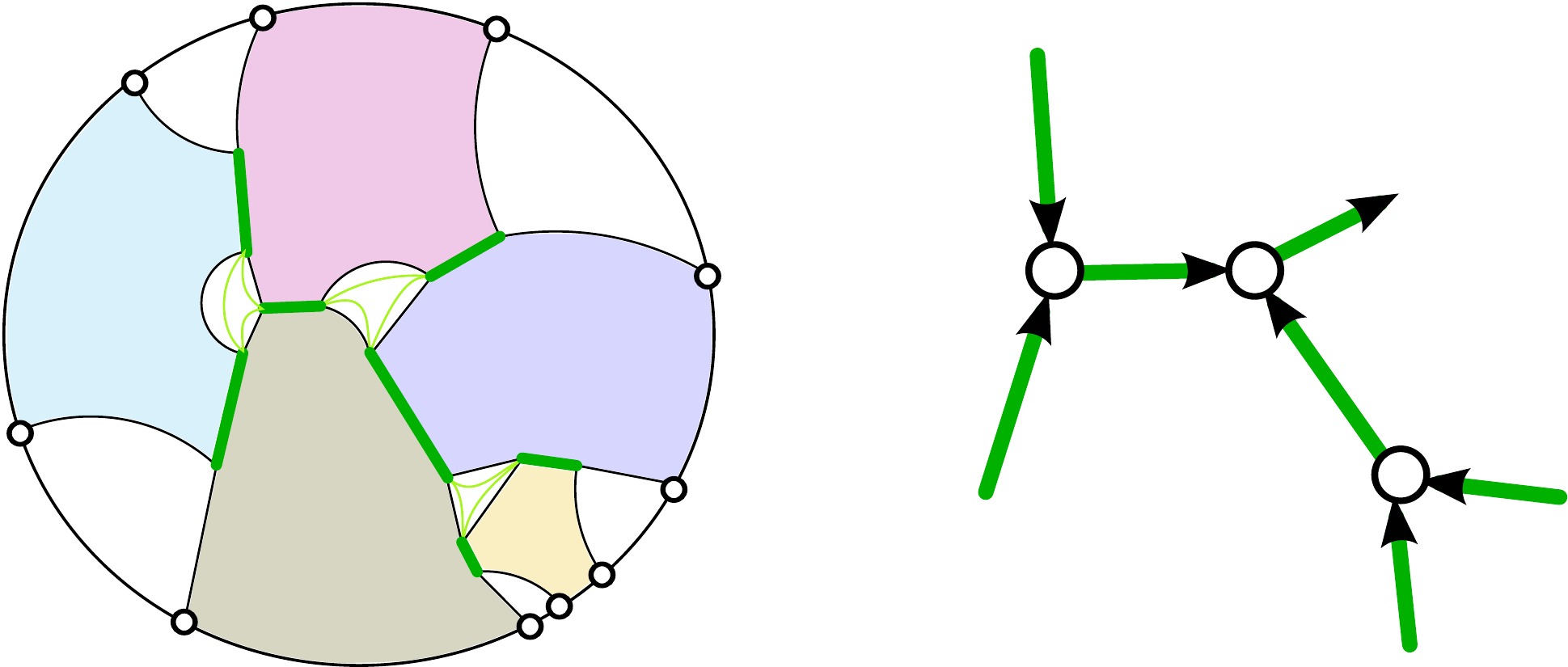}
      \caption{(left) Thick green segments are borders, thin green curves are in $\mathbb{R}^2\setminus (F\cup \mathcal{R})$ connecting borders that are incident in $T$. (right) The directed tree $T$ used for dynamic programming.}
\label{fig:T}
\end{figure}

By the correspondence of the vertices of $T$ to non-regions, we have:
\begin{observation}\label{obs:region-path-T}
  The borders $B_{i,j}, B_{i,\ell}, \ldots$ along a given region $R_i$ form a path in $T$.  The order of the borders from $s_i$ to $t_i$ along $P_i$ is the same as in the path in $T$.
\end{observation}
Consider two edges $B_{i,j}$ and $B_{h,\ell}$ that are incident to the same vertex $v$ of $T$ and that are consecutive in a cyclic order of the edges incident to $v$ in $T$'s embedding.  By Observation~\ref{obs:region-path-T} and the embedding of $T$, there is a labeling of $i,j,h,\ell$ such that:
\begin{observation}\label{obs:first-last-border}
  Either $j = h$ or $B_{i,j}$ is the last border of $R_j$ along $P_j$ and $B_{h,\ell}$ is the first border of $R_h$ along $P_h$.
\end{observation}


Root $T$ at an arbitrary leaf.  Since $\cal R$ is connected, the leaf of $T$ is non-trivial; that is, it has an edge $B_{i,j}$ incident to it.  By Observation~\ref{obs:region-path-T}, $B_{i,j}$ is (w.l.o.g.) the last border of $R_i$ along $P_i$ and the first border of $R_j$ along $P_j$.  By the correspondence of the vertices of $T$ to non-regions, and the connectivity of $\cal R$, either $j = i+1$ or $i = 1$ and $j = k$.
For ease of notation, we assume that the terminals are numbered so that $i = 1$ and $j = k$.  We get:
\begin{observation}\label{obs:leaf-edge}
  Every non-root leaf edge of $T$ corresponds to a border $B_{i,i+1}$.
\end{observation}

We consider the borders to be both paths in $G$ and edges in $T$. In $T$ we orient the borders toward the root.  In $G$, this gives a well defined start $a_{i,j}$ and endpoint $b_{i,j}$ of the corresponding path $B_{i,j}$ (note that $a_{i,j} = b_{i,j}$ is possible).  By our choice of terminal numbering and orientation of the edges of $T$, from $s_i$ to $t_i$ along $\spath{i}$, $b_{i,j}$ is visited before $a_{i,j}$, and from $s_j$ to $t_j$ along $\spath{j}$, $a_{i,j}$ is visited before $b_{i,j}$.

\subsection{Dynamic programming table}
We populate a dynamic programming table $D_{i,j}$ for each border $\border{i}{j}$.  $D_{i,j}$ is indexed by two vertices $x$ and $y$: $x$ is a vertex of $\spath{i}[b_{i,j},t_i]$  and $y$ is a vertex of $\spath{j}[s_j,b_{i,j}]$.  $D_{i,j}[x,y]$ is defined to be the minimum length of a set of vertex-disjoint paths that connect:
  \begin{center}
    $x$ to $t_i$, $s_j$ to $y$, and $s_h$ to $t_h$ for every $h = i+1,
    \ldots, j-1$
  \end{center}
These paths are illustrated in Figure~\ref{fig:dp-table}.
We interpret $y$ as the last vertex of $\spath{j}[s_j,b_{i,j}]$ that is used in this sub-solution and we interpret $x$ as the first vertex of $\spath{i}[b_{i,j},t_i]$ that can be used in this sub-solution (or, more intuitively, the last vertex of the reverse of $\spath{i}[b_{i,j},t_i]$).  By Lemma~\ref{lem:sp_in_region}, each of the paths defining $D_{i,j}[x,y]$ are contained by their respective region.
\begin{figure}[ht]
  \centering
    \includegraphics[height=1.4in]{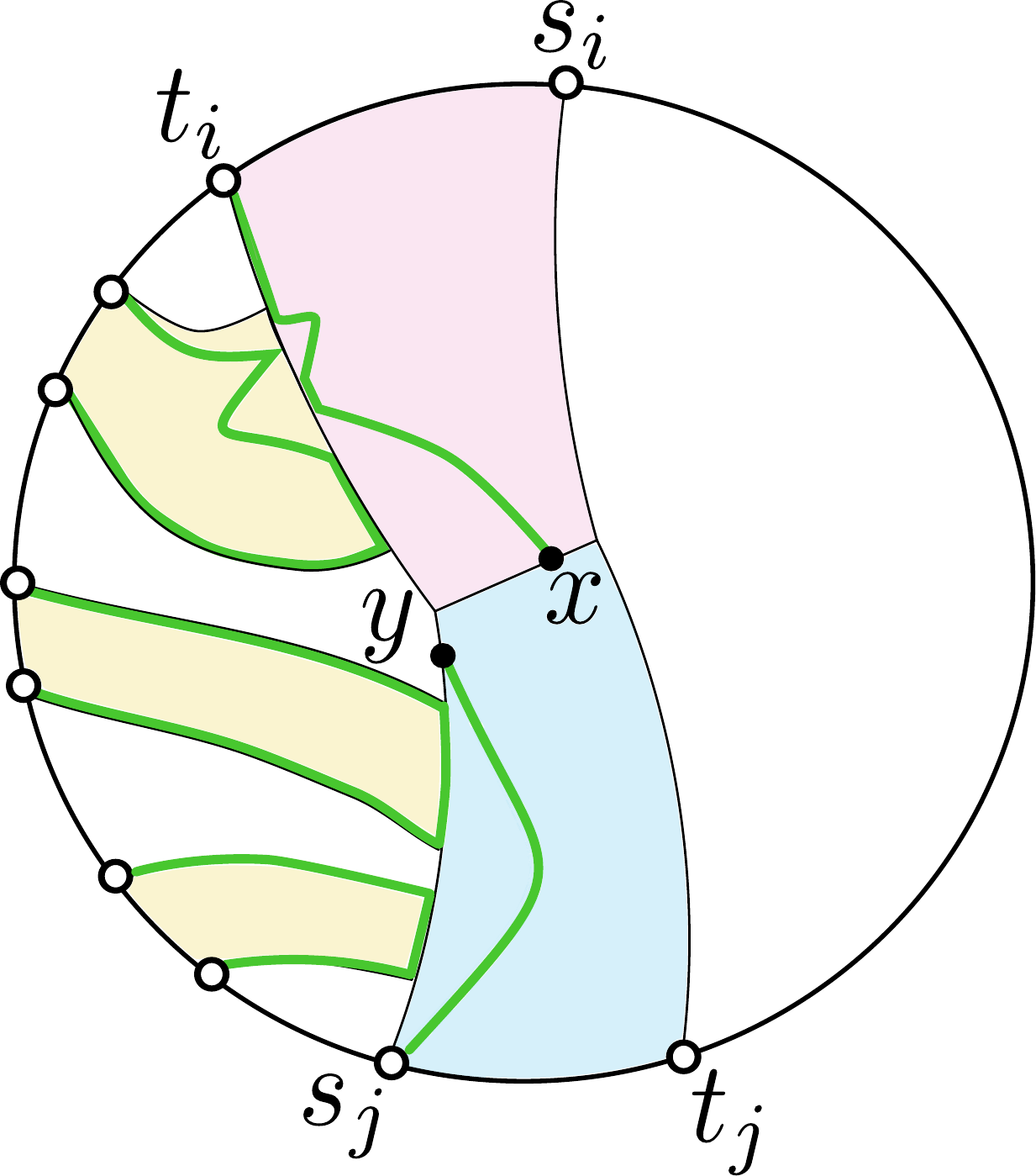}
  \caption{An illustration of the paths defined by $D_{ij}[x,y]$.}
  \label{fig:dp-table}
\end{figure}

\paragraph{Optimal solution}  Given $D_{1,k}$, we can compute the value of the optimal solution.  By Lemma~\ref{lem:alpha_opt}, $\opt{1}$ and $\opt{k}$ contain a shortest (possibly trivial) path from $s_1$ to a vertex $x$ on $\spath{1}$, and from a vertex $y$ on $\spath{k}$ to $t_k$, respectively.  Let $y$ be the last vertex of $\spath{k}[s_k,b_{1,k}]$ that $\opt{k}$ contains and let $x$ be the first vertex of $P_1[b_{1,k},t_1]$ that $\opt{1}$ contains.  Then, by Lemma~\ref{lem:alpha_opt}, the optimal solution has length $D_{1,k}[x,y] + \len{\alpha(s_1,x)} + \len{\alpha(y,t_k)}$ where $\alpha$ is a Type~I path between the given vertices.  The optimal solution can be computed in $O(n^2)$ time 
by enumerating over all choices of $x$ and $y$.  Computing all such Type~I paths takes $O(n^2)$ since there are 
$O(n)$ such paths to compute, and each path can be found using the linear-time shortest paths algorithm for planar graphs~\cite{henzinger1997planarShortestPaths}.

\paragraph{Base case: Leaf edges of $\mathbf T$}  Consider a non-root leaf edge of $T$, which, by Observation~\ref{obs:leaf-edge}, is $B_{i,i+1}$ for some $i$.  Then $D_{i,i+1}[x,y]$ is the length of minimum vertex disjoint $x$-to-$t_i$ and $s_{i+1}$-to-$y$ paths in $\region{i} \cup \region{i+1}$.
By Lemma~\ref{lem:2minSumplus}, $D_{i,i+1}[x,y]$ can be computed in $O(n)$ time  for any $x$ and $y$ and so $D_{i,i+1}$ can be populated in $O(n^2)$ time. 

\subsection{Non-base case of the dynamic program}

Consider a border $\border{i}{j}$ and consider the edges of $T$ that are children of $\border{i}{j}$.  These edges considered counter-clockwise around their common node of $T$ correspond to borders $\border{i_1}{j_1},\border{i_2}{j_2},\ldots,\border{i_t}{j_t}$ where $i \leq i_1 \leq j_1 \leq \cdots \leq i_t \leq j_t \leq j$.  For simplicity of notation, we additionally let $j_0 = i$ and $i_{t+1} = j$.
Then, by Observation~\ref{obs:first-last-border}, either $j_\ell = i_{\ell+1}$ or $\border{i_\ell}{j_\ell}$ is the last border on $\spath{j_\ell}$ and $\border{i_{\ell+1}}{j_{\ell+1}}$ is the first border on $\spath{i_{\ell+1}}$ for $\ell = 0,\ldots,t$.

\paragraph{An acyclic graph $H$ to piece together sub-solutions}
To populate $D_{i,j}$ we create a directed acyclic graph $H$ with sources corresponding to vertices of $\spath{i}[b_{i,j},t_i]$ and sinks corresponding to $\spath{j}[s_j, b_{i,j}]$.  A source-to-sink ($u$-to-$v$) path in $H$ will correspond one-to-one with vertex disjoint paths from:
\begin{center}
  $u$ to $t_i$, $s_j$ to $v$, and $s_{h}$ to $t_{h}$ for every $h = i+1,
  \ldots, j-1$
\end{center}
Here $u$ and $v$ do not correspond to the vertices $x$ and $y$ that index $D_{i,j}$; to these vertex disjoint paths, we will need to append vertex disjoint $x$-to-$u$ and $v$-to-$y$ paths (which can be found using a minimum cost flow computation by Lemma~\ref{lem:2minSumplus}).

The arcs of $H$ are of two types: (a) Type~I arcs and (b) sub-problem arcs.  Directed paths in $H$ alternate between these two types of arcs.  The Type~I arcs correspond to Type~I paths and the endpoints of the Type~I arcs correspond to the endpoints of the Type~I paths.
Sub-problem arcs correspond to the sub-solutions from the dynamic programming table and the endpoints of the sub-problem arcs correspond to the indices of the dynamic programming table (and so are the endpoints of the {\em incomplete} paths represented by the table).  Note that vertices of a border may appear as either the first or second index to the dynamic programming table; in $H$, two copies of the border vertices are included so the endpoints of the resulting sub-solution arcs are distinct.  Formally:
\begin{itemize}
\item Type~I arcs go from vertices of $\spath{j_\ell}$ to vertices of $\spath{i_{\ell+1}}$ for $\ell = 0,\ldots,t$.  Consider regions $R_{j_\ell}$ and $R_{i_{\ell+1}}$.  There are two cases depending on whether or not $R_{j_\ell} = R_{i_{\ell+1}}$.
  \begin{itemize}
  \item If $R_{j_\ell} = R_{i_{\ell+1}}$, then for every vertex $x$ of $\spath{j_\ell}[s_{j_\ell},b_{i_\ell,j_\ell}]$ and every vertex $y$ of $\spath{j_\ell}[b_{i_{\ell+1},j_{\ell+1}},t_{j_\ell}]$, we define a Type~I arc from $x$ to $y$ with length equal to the length of the $x$-to-$y$ Type~I path.
  \item If $R_{j_\ell} \ne R_{i_{\ell+1}}$, then for every vertex $x$ of $\spath{j_\ell}[s_{j_\ell},b_{i_\ell,j_\ell}]$ and every vertex $y$ of $\spath{i_{\ell+1}}[b_{i_{\ell+1},j_{\ell+1}},t_{j_\ell}]$, we define a Type~I arc from $x$ to $y$ with length equal to the sum of the lengths of the $x$-to-$t_{j_\ell}$ and $s_{i_{\ell+1}}$-to-$y$ Type~I paths.
  \end{itemize}
\item Sub-problem arcs go from vertices of $\spath{i_\ell}$ to vertices of $\spath{j_\ell}$ for $\ell = 1, \ldots, t$. For every $\ell = 1, \ldots, t$ and every vertex $x$ of $\spath{i_\ell}$ and vertex $y$ of $\spath{j_\ell}$ (that are not duplicates of each other), we define a sub-problem arc from $x$ to $y$ with length equal to $D_{i_\ell,j_\ell}[x,y]$.
\end{itemize}

\paragraph{Shortest paths in $H$}
By construction of $H$ and the definition of $D_{i_\ell,j_\ell}$, for a source $u$ and sink $v$, we have:
\begin{observation}
  There is a $u$-to-$v$ path in $H$ with length $L$ if and only if there are vertex disjoint paths of total length $L$ from $u$ to $t_i$, $s_j$ to $v$, and $s_{h}$ to $t_{h}$ for every $h = i+1,
  \ldots, j-1$.
\end{observation}
See Figure~\ref{fig:one-to-one} for an illustration of the paths in $G$ that correspond to a source-to-sink path in $H$. Let $H[u,v]$ denote the shortest $u$-to-$v$ path in $H$ (for a source $u$ and a sink $v$).  We will need to compute $H[u,v]$ for every pair of sources and sinks.  Since every vertex in $G$ appears at most twice in $H$, the size of $H$ is $O(n^2)$ and for a given sink and for all sources, the shortest source-to-sink paths can be found in time linear in the size of $H$ using dynamic programming.  Repeating for all sinks results in an $O(n^3)$ running time to compute $H[u,v]$ for every pair of sources and sinks.\footnote{Computing the length of the Type~I paths is dominated by $O(n^3)$, but can be improved to $O(n\log n)$ time by running Klein's boundary shortest path algorithm~\cite{klein2005multiSourceShPaths} in all regions, resulting in an $O(n^2)$ time to construct $H$.}

\begin{figure}[ht]
  \centering
  \includegraphics[height=1.7in]{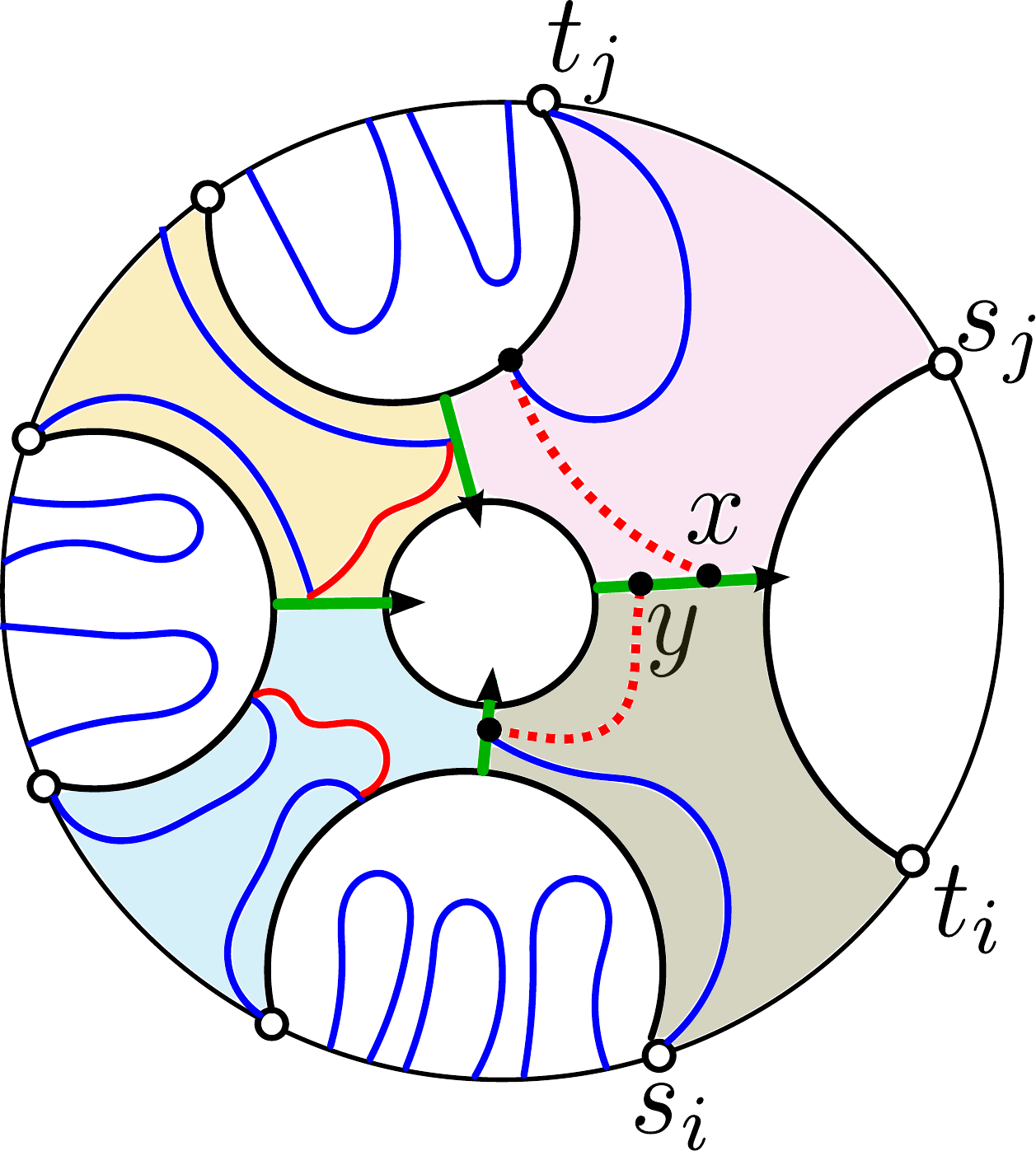}
  \caption{Mutually disjoint walks represented by a directed path in $H$ for a set of incident borders (green).  The blue arcs correspond to the walks represented in sub-problems and the solid red paths correspond to the Type~I paths represented by Type~I arcs.  The dotted red paths represent vertex disjoint $u$-to-$x$ and $v$-to-$y$ paths that will be added via a min-cost flow computation.}
\label{fig:one-to-one}
\end{figure}

\paragraph{Handling vertices that appear in more than two regions}  As indicated, a vertex $c$ may appear in more than two regions; this occurs when two or more borders share an endpoint.  In the construction above, if $c$ appears in only two regions, then, $c$ can only be used as the endpoint of two sub-paths (whose endpoints meet to form a part of an $s_i$-to-$t_i$ path in the global solution).  However, suppose for example that $c$ appears as an endpoint of both $\border{i}{j}$ and $\border{i'}{j'}$ and so 4 copies of $c$ are included in $H$ (two copies for each of these borders).  On the other hand, one need only {\em guess} which $s_i$-to-$t_i$ path $c$ should belong to first and construct $H$ accordingly.  There are only $k$ possibilities to try.

Unfortunately, there may be $O(k)$ shared vertices among the borders \\ $\border{i_1}{j_2},\border{i_2}{j_2},\ldots,\border{i_t}{j_t}$ involved in populating $D_{i,j}$.
  It seems that for each of these $O(k)$ shared vertices, one would need to guess which $s_i$-to-$t_i$ path it belongs to, resulting in an exponential dependence on $k$. 

Here we recall the structure of $T$: the nodes of $T$ correspond to {\em non-regions}: disks (or points) surrounded by regions.  If there are several shared vertices among the borders, then there is an order of these vertices around the boundary of the non-region.  That is, for a vertex shared by a set of borders, these borders must be contiguous subsets of $\border{i_1}{j_2},\border{i_2}{j_2},\ldots,\border{i_t}{j_t}$.  In terms of the construction of $H$, there is a contiguous set of levels that a given shared vertex appears in and distinct shared vertices participate in non-overlapping sets of levels.  For one set of these levels, we can create different copies of the corresponding section of $H$.  In each copy we modify the directed graph to reflect which $s_i$-to-$t_i$ path the corresponding shared vertex may belong to (see Figure~\ref{fig:copies}).  As we have argued, since distinct shared vertices participate in non-overlapping sets of levels, this may safely be repeated for every shared vertex.  The resulting graph has size $O(kn^2)$ since there are $O(k)$ borders and shared vertices are shared by borders.  The resulting running time for computing all source-to-sink shortest paths in the resulting graph is then $O(kn^3)$.

\begin{figure}[ht]
  \centering
  \includegraphics[height=4in]{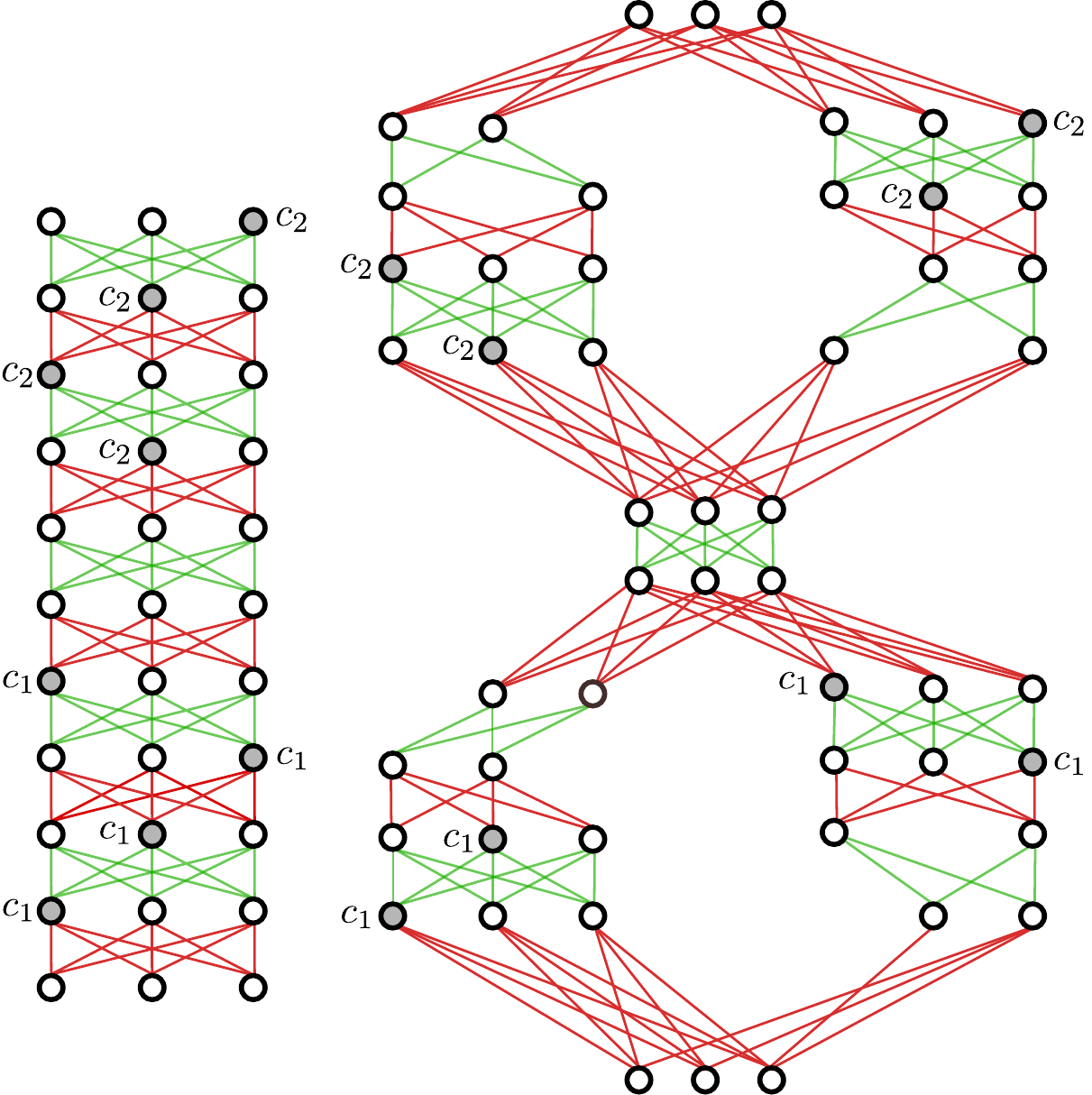}
  \caption{(left) $H$ constructed without handling the fact that vertices $c_1$ and $c_2$ (black) may appear more than twice.  The arcs are all directed upwards (arrows are not shown); green arcs are sub-problem arcs and red arcs are Type~I arcs. (right) The levels that $c_1$ appears in are duplicated, and only one pair of copies of $c$ is kept in each copy of these levels. The vertex $c_1$ may only be visited twice now on a source to sink path.}
\label{fig:copies}
\end{figure}

\paragraph{Computing $D_{i,j}$ from $H$}
To compute $D_{i,j}[x,y]$, we consider all possible $u$ on $\spath{i}[x,t_i]$ and $v$ on $\spath{j}[s_j,y]$ and compute the minimum-length vertex disjoint $u$-to-$x$ path and $v$-to-$y$ path that only use vertices that are interior to $\region{i} \cup \region{j}$ (that is vertices of $\border{i}{j}$ may be used); by Lemma~\ref{lem:2minSumplus}, these paths can be computed in linear time.  Let $M[u,v]$ be the cost of these paths.  Then
\[
D_{i,j}[x,y] = \min_{u \in \spath{i}[x,t_i],\, v \in \spath{j}[s_j,y]} M[u,v]+H[u,v].
\]
As there are $O(n^2)$ choices for $u$ and $v$ and $M[u,v]$ can be computed in linear time, $D_{i,j}[x,y]$ can be computed in $O(n^3)$ time given that distances in $H$ have been computed.

\paragraph{Overall running time}  For each border $\border{i}{j}$, $H$ is constructed and shortest source-to-sink paths are computed in $O(k n^3)$ time.  For each $x,y \in \border{i}{j}$, $D_{i,j}[x,y]$ is computed in $O(n^3)$ time.  Since there are $O(n^2)$ pairs of vertices in $\border{i}{j}$, $D_{i,j}$ is computed in $O(n^5)$ time (dominating the time to construct and compute shortest paths in $H$).  Since there are $O(k)$ borders (Lemma~\ref{lem:bounded_number_borders}),  the overall time for the dynamic program is $O(kn^5)$.

\paragraph{Acknowledgements}  This material is based upon work supported by the National Science Foundation under Grant No.\ CCF-1252833.

\bibliographystyle{alpha}
\bibliography{dispaths}

\end{document}